\documentclass{article}

\usepackage[margin=1.15in]{geometry}
\usepackage{amssymb} 
\usepackage{amsthm}
\usepackage{graphicx} 
\usepackage{amsmath}
\usepackage{subfigure} 
\usepackage{todonotes}
\usepackage{placeins}

\usepackage{algorithm} 
\usepackage{algorithmicx}
\usepackage{algpseudocode} 

\usepackage{url}

\newcounter{maincounter} 
\newtheorem{theorem}{Theorem}
\newtheorem{definition}[maincounter]{Definition}

\newtheorem{lemma}[maincounter]{Lemma} 
\newtheorem{corollary}[maincounter]{Corollary}

\newtheorem{fact}{Fact}

\newcommand{\ubar}[1]{\mkern3.5mu\underline{\mkern-3.5mu
#1\mkern-3.5mu}\mkern3.5mu }
\renewcommand{\O}{\ensuremath{\mathcal{O}}}

\newcommand{\R}{\ensuremath{\operatorname{R}}}
\newcommand{\Rmax}{\ensuremath{\bar{\R}}}
\newcommand{\Rmin}{\ensuremath{\ubar{\R}}}
\newcommand{\amax}{\ensuremath{\bar{\alpha}}}
\newcommand{\amin}{\ensuremath{\ubar{\alpha}}}
\newcommand{\noise}{{\operatorname{N}}} 
\newcommand{\nmax}{\bar{{\operatorname{N}}}}

\newcommand{\bmax}{\bar{\beta}}

\newcommand{\dist}{{\operatorname{dist}}}
\newcommand{\area}{{\operatorname{Area}}}

\newcommand{\power}{{\operatorname{P}}}

\begin{document}

\title{Arbitrary Transmission Power in the SINR Model: \\ Local
  Broadcasting, Coloring and MIS}

\author{
\vspace{2mm}
  Fabian Fuchs and  Dorothea Wagner\\
  {\normalsize Karlsruhe Institute for Technology } \\
  {\normalsize   Karlsruhe, Germany } \\
  {\normalsize   \{fabian.fuchs, dorothea.wagner\}@kit.edu } 
}

\maketitle

\noindent\makebox[\textwidth][c]{%
\begin{minipage}[t]{0.8\textwidth}
\begin{abstract}
  In the light of energy conservation and the expansion of existing
  networks, wireless networks face the challenge of nodes with
  heterogeneous transmission power.  However, for more realistic
  models of wireless communication only few algorithmic results are
  known.  In this paper we consider nodes with arbitrary, possibly
  variable, transmission power in the so-called physical or SINR
  model.  Our first result is a bound on the probabilistic
  interference from all simultaneously transmitting nodes on
  receivers. This result implies that current local broadcasting
  algorithms can be generalized to the case of non-uniform
  transmission power with minor changes. The algorithms run in
  $\O(\Gamma^{2} \Delta \log n)$ time slots if the maximal degree
  $\Delta$ is known, and $\O((\Delta + \log n)\Gamma^{2} \log n)$
  otherwise, where $\Gamma$ is the ratio between the maximal and the
  minimal transmission range.  The broad applicability of our result
  on bounding the interference is further highlighted, by generalizing
  a distributed coloring algorithm to this setting.
\end{abstract}
\end{minipage} 
}

\newpage

\section{Introduction}
\label{sec:introduction}

One of the most fundamental problems in wireless ad hoc networks is to
enable efficient communication between neighboring nodes.  This
problem recently received increasing attention among the distributed
algorithm community, as more refined models of wireless communication
became established in algorithms research.  Among these models, the
so-called physical or \emph{signal-to-interference-and-noise} (SINR)
model is most prominent and promising, due to its common use in the
engineering literature.  However, so far most algorithmic work in the
SINR model is restricted to the case of uniform transmission power. In
this case, \emph{local broadcasting}
\cite{hm-ttblb-12,yhwl-aodaa-12,ywhl-dlbap-11,gmw-lbpim-08} provides
initial communication by enabling each node to transmit one message
such that all intended receivers (i.e., neighbors) are able to decode
the message.
% Other frequently considered problems are distributed node coloring
% \cite{dt-dncsi-10,ywhl-ddcpm-11}, link \cite{mwz-tcmsi-06} or transmission scheduling
% \cite{jk-dbsas-12,fw-olbsc-13} and  connectivity
% \cite{mw-tccwn-06,alpp-anupc-09}.

In this work we consider the problem of local broadcasting in the SINR
model under arbitrary transmission power assignment, i.e., each node
has its individual, possibly variable, transmission power. We are the
first to consider this setting from an algorithmic perspective.  While
some distributed node coloring algorithms do consider the transmission
power to be variable \cite{dt-dncsi-10,ywhl-ddcpm-11}, they still
increase the transmission power synchronously and thus effectively
operate on an uniform power network.  The sole line of research that
leverages non-uniform transmission power is on link scheduling and
capacity maximization \cite{hm-wcc-12,kv-dcr-10}. However, there, each
node is usually considered to be either transmitter or receiver. If a
node has multiple roles it might have to adapt its transmission power
frequently.  On the other hand, the effects of heterogeneous
transmission power are considered in simulation-based studies for
example in \cite{ftbus-amac-02,pkd-macn-01}, %
while the case of unidirectional communication links, which are a
result of heterogeneous transmission powers, is studied even more
frequently \cite{6266131,wang2008mac}.

We assume the harsh environment of an wireless ad hoc network just
after deployment. In particular, we consider multi-hop networks, where
the nodes do initially not have any information about whether other
nodes are awake, have already started the algorithm or in which phase
of the algorithm they are. The only knowledge they may have is an
upper bound on the number of neighbors, and a rough bound on the total
number of nodes in the network.  Note that our model does not assume a
\emph{collision detection} mechanism.  Additionally to this harsh
model, we also considered some recent ideas regarding practical
matters of algorithms for wireless networks by Kuhn \emph{et. al.}
\cite{dgkn-bahs-13}. They promoted the use of lower and upper bounds
for important network parameters such as $\alpha, \beta$ and $\noise$
(cf. Section~\ref{sec:preliminaries}). This is an important step
towards practicability of the algorithms as upper and lower bounds to
these values are well-represented in the literature, however, exact values
vary depending on the network environment.

\subsection{Contributions}
\label{sec:contrib}

In this work we are the first to consider arbitrary transmission
powers in the SINR model, and thus networks with unidirectional links
for the problems of local broadcasting, distributed node coloring and
MIS. However, our first contribution is of more general nature and
provides an abstract method for bounding the interference in these
networks. We prove that transmissions are feasible based on the sum of
local transmission probabilities. This result is widely applicable, as
verifying that the sum of local transmission probabilities is bounded
as required, is relatively simple. % Our algorithmic results in this
% paper use this bound to ensure successful communication.

Our second result transfers algorithms for local broadcasting
presented in \cite{hm-ttblb-12,gmw-lbpim-08} to the case of arbitrary
transmission power assignment. We achieve local broadcasting in
$\O(\Gamma^{2} \Delta \log n)$ time slots if the maximal degree
$\Delta$ is known and $\O((\Delta + \log n)\Gamma^{2} \log n)$
otherwise, where $\Gamma$ is the ratio between the maximal and the
minimal transmission range. Note that these bounds match those for the
uniform case if the algorithms are run on such networks.
Additionally we discuss the case of variable transmission
power in Section~\ref{sec:vari-transm-power}, which achieves similar
bounds, but allows nodes to change the transmission power in each time
slot instead of fixing it for each round of local broadcasting.

Finally we give an algorithm for distributed node coloring in these
harsh environments. The algorithm is in based on an algorithm by
Moscibroda and Wattenhofer \cite{mw-curn-08}, which was adapted to the
uniform SINR model by Derbel and Talbi \cite{dt-dncsi-10}. Note
however, that fundamental changes to the algorithm itself are required
due to the increased complexity of the network structure, such as
unidirectional communication links. We introduce a new network
parameter $\ell$, that measures the length of the longest simple
unidirectional chain in the partially directed network and prove that
our distributed node coloring algorithm colors the network with
$\O(\Gamma^2\Delta)$ colors in $\O((\Delta+\ell) \Gamma^{6} \Delta
\log n)$ time slots. By simplifying the algorithm we obtain an
algorithm that computes an MIS in $\O(\ell \Gamma^{4} \log n)$ time
slots.
Note that all our algorithms are fully operational in the unstructured
radio network, especially under asynchronous node wake-up and
sleep. 

\subsection{Related Work}
\label{sec:related-work}

The study of local broadcasting, and interference in general, has only
recently emerged. Especially in classical distributed message
passing models such as $\mathcal{LOCAL}$ or $\mathcal{CONGEST}$
\cite{p-dcals-00}, the transmission of a message to neighbors is
guaranteed. However, this is not the case for wireless networks. Hence
interference in general and local broadcasting in particular must be
considered in the more realistic SINR model of
interference. Goussevskaia \emph{et. al.}  \cite{gmw-lbpim-08} were
the first to present local broadcasting algorithms in the SINR
model. Their first algorithm assumes an upper bound $\Delta$ on the
number of neighbors to be known by the nodes and solves local
broadcasting with high probability in $\O(\Delta \log n)$ time, while
the second algorithm does not assume this knowledge and requires
$\O(\Delta \log^3n)$ time.  The second algorithm has subsequently been
improved by Yu \emph{et. al.} to run in $\O(\Delta \log^2n)$
\cite{ywhl-dlbap-11}, and again to $\O(\Delta\log n + \log^2 n)$
\cite{yhwl-aodaa-12}.  This bound has been matched by Halld\'{o}rsson
and Mitra in \cite{hm-ttblb-12} using a more robust algorithm, along
with an algorithm that leverages carrier sensing to achieve a time
complexity of $\O(\Delta + \log n)$. 

Research on distributed node coloring dates back to the first days of
distributed computing nearly 30 years ago.  Due to the wide variety of
results in this area, we refer to the monograph recently published by
Barenboim and Elkin \cite{be-dcg-13} for results in the
$\mathcal{LOCAL}$ model.  Note that the considered message passing
model abstracts away characteristics of a newly deployed wireless ad
hoc network: Global interference, asynchronous node wake-up and sleep,
and unidirectional communication links are not considered. Thus these
algorithms cannot directly be used in the harsh model considered in
this work.

An algorithm that colors the network with $\O(\Delta)$ colors in
$\O(\Delta \log n)$ time was presented by Moscibroda and Wattenhofer
in \cite{mw-curn-05}. However, they assume a graph-based interference
model. The algorithm has subsequently been improved in
\cite{mw-curn-08} and \cite{sw-cuwm-09} and transfered to the SINR
model by Derbel and Talbi \cite{dt-dncsi-10} with the same bound on
colors and runtime as the original algorithm. Yu \emph{et. al.}
consider the problem of coloring with only $\Delta +1$ colors in
\cite{ywhl-ddcpm-11} and present algorithms that run in $\O(\Delta
\log^2 n)$ time slots or $\O(\Delta \log n + \log^2 n)$ if the nodes
transmission power can be tuned by a constant factor.

\section{Preliminaries}
\label{sec:preliminaries}

We consider a wireless network consisting of $n$ nodes, that are
placed arbitrarily on the Euclidean plane. We assume that all nodes in
the network know their ID and an upper bound $\tilde{n}$ on $n$, with
$\tilde{n} \leq n^c$ for some constant $c \geq 1$. As the upper bound
influences our results only by a constant factor we usually write $n$
even though only $\tilde{n}$ may be known by the nodes. Also, we
assume that nodes know lower and upper bounds on the transmission
power or the transmission ranges. This assumption is realistic, as
lower bounds for reasonable minimal transmission ranges can be
computed while upper bounds (for specified frequencies) are often
regulated by public authorities.

In the geometric SINR model a transmission from node $v$ to node $w$
is successful iff the SINR condition holds:
\begin{equation}
  \frac{ \frac{ \power_v }{ \dist(v,w)^\alpha } } { \sum_{ u \in
      \mathcal{I} } \frac{ \power_u }{ \dist(u,w)^\alpha } + \noise} \geq \beta
\end{equation}
where $\power_v$ ($\power_u$) denotes the transmission power of node
$v$ ($u$), $\alpha$ is the attenuation coefficient, which depends on
the environment and characterizes how fast the signal fades. The
SINR-threshold $\beta \geq 1$ is a hardware-defined constant, $\noise$ is
the environmental noise and $\mathcal{I}$ is the set of nodes transmitting
simultaneously with $v$.
As introduced in \cite{dgkn-bahs-13} and motivated by the hardness of
determining exact network parameters we restrict our nodes knowledge
to upper and lower bounds of the values $\alpha$, $\beta$ and $\noise$
and denote them by e.g. $\amin$ and $\amax$ for the minimal and
maximal values.

Based on the SINR constraints, we define the \emph{maximum
  transmission range} of a node $v$ to be $\Rmax_v =
(\frac{\power_v}{\nmax \bmax})^{1/\amax}$.  Note that this is maximal
under the restriction that this range can be reached regardless of the
actual network parameters $\alpha$, $\beta$, $\noise$.  The global
maximum transmission range in the network is denoted by $\Rmax$, the
minimum range by $\Rmin$ and the ratio between $\Rmax$ and $\Rmin$ by
$\Gamma = \frac{\Rmax}{\Rmin}$.  Due to the SINR constraints, a node
$v$ cannot reach another node $w$ which is located at the maximum
transmission range of $v$, as soon $v$ transmits simultaneously with
any other node in the network.  As having only one simultaneous
transmission in the network is not desired, we use a parameter $\delta
> 1$ to determine the distance up to which the nodes messages should
be received.  We call this distance the \emph{broadcasting range}
$\R_v = ( \frac{ \power_v }{ \delta \nmax \bmax })^{1/\amin}$ and the
region within this range from $v$ the broadcasting region $B_v$.  We
denote the maximum number of nodes within the transmission range
$\Rmax_v$ of any $v$ as $\Delta$. This is an upper bound on the number
of nodes reachable from $v$, since the broadcasting range $\R_v$ is
fully contained in the transmission range.  Note that $\Delta$ is
known by the nodes only if stated with the corresponding
algorithms. We define the \emph{proximity region} around $v$ as the
area closer than $3\Rmax$ to $v$.  Note that even though we use time
slots in our analysis, we do not require a global clock or
synchronized time slots in our algorithm. Decent local clocks are
sufficient, while time slots are only required in the analysis.

\textbf{Roadmap:} In the following section we bound the
probabilistic interference of nodes outside the proximity region based
on the sum of transmission probabilities from within each transmission
region. In Section~\ref{sec:local-broadcasting} we apply this result
to previous results on local broadcasting and thereby transfer current
algorithms to the more general model. 
% Our results regarding variable
% transmission power are presented in
% Section~\ref{sec:vari-transm-power}. 
The applicability of our results is highlighted in
Section~\ref{sec:node-coloring-mis}, as we consider the problem of
distributed node coloring and generalize a well-known algorithm from
the case of uniform transmission powers. We conclude this paper in
Section~\ref{sec:conclussion} with some final remarks.

% Additionally to the realistic SINR constraints, we assume the harsh
% conditions of the unstructured network
% model which allows the nodes to turn on and off asynchronously and
% thereby model the environmental restrictions of a network that 
% We assume synchronous rounds, i.e., time is divided in rounds, in
% every round each node may perform computations, transmit or receive a
% message (according to the SINR constraints). Note that it is possible
% to run an algorithm based on synchronous rounds in an asynchronous
% system with the same asymptotic complexity by adding some
% synchronization messages (cite Complexity of network synchronization
% by B. Awerbuch [this is what Kuhn does in a PODC09 paper]).

Note that even though we use time slots in our analysis, we do not
require a global clock or synchronized time slots in our
algorithm. Decent local clocks are sufficient, while time slots are
only required in the analysis. 
% As shown in \cite{tk-psrc-75-part2}, it
% is justified to assume slotted transmission in the analysis since
% slotted vs. unslotted Aloha differ only by a factor of 2.

\noindent
\textbf{Roadmap:} In the following section we will bound the
probabilistic interference of nodes outside the proximity region based
on a bound on the sum of transmission probabilities from within each
transmission region. In Section~\ref{sec:local-broadcasting} we apply
this result to previous results on local broadcasting and thereby
transfer current algorithms to the more general model. In
Section~\ref{sec:node-coloring-mis} we consider distributed node
coloring and describe an algorithm that is capable of computing an
$\O(\Gamma^2\Delta)$ coloring, or after a simplification an MIS. We
conclude this paper in Section~\ref{sec:conclussion} with some final
remarks.

\section{Bounding the Interference}
\label{sec:bound-interf}

In contrast to other models for interference in wireless communication
such as the protocol model, the SINR model captures the global aspect
of interference and reflects that even interference from far-away
nodes can add up to a level that prevents the reception of
transmissions from relatively close nodes.  To ensure that a given
transmission can be decoded by all nodes within the broadcasting
range, one usually proves that reception within a certain time
interval is successful \emph{with high probability} (w.h.p.---with
probability at least $1-\frac{1}{n^c}$ for a constant $c > 1$). Such a
proof can be split in two parts
\begin{enumerate}
\item The probability that a node transmits within a proximity region
  around a sender is constant
\item Let $P_{\text{2high}}(v)$ be the event that the interference
  from all nodes outside of the proximity region of $v$ on nodes in
  the broadcasting region of $v$ is too high. Show that
  $P_{\text{2high}}(v)$ has constant probability.
\end{enumerate}
We shall follow this scheme by considering the transmission of an
arbitrary node and proving that both conditions hold with constant
probability in each time slot, and hence a local broadcast is
successful with high probability.

In order to make the result general and applicable to many different
settings, we make only one very general assumption. Namely we assume
the sum of transmission probabilities from within a broadcasting
region to be bounded by a constant. This is very common and allows us
to apply the analysis from this section in the following Sections
\ref{sec:local-broadcasting} and \ref{sec:node-coloring-mis} to
generalize algorithms designed for the uniform transmission power case
to the more general case considered in this paper\footnote{We can directly
  apply our results to many algorithmic results in the SINR model,
  however the algorithms themselves often rely on bidirectional communication
  links.}.
\begin{definition}
  \label{def:1}
  Given a network of $n$ nodes with at most $\Delta$ nodes in each
  transmission region. Let $\gamma$ be the upper bound on the sum of
  transmission probabilities from within one transmission region.
\end{definition}
Let the upper bound on the sum of transmission probabilities from
within each transmission region be
\begin{align}
  \label{eq:1}
  \gamma := \frac{(\delta-1)}{120 \bmax \Gamma^{2} \sum_{i=1}^n
    \frac{1}{i^{\amax-1}}}.
\end{align}
Note that this bound can be realized, for example by requiring nodes
to transmit with probability $\gamma/\Delta$. Another option is the
so-called slow-start technique,
cf. Section~\ref{sec:lb-without-knowl-delta}.  The constant is of the
stated form, mainly to bound the interference from all other nodes in
the network in the proof of
Theorem~\ref{thm:sum-of-rings-bounded-interference}. It holds that
$\gamma \leq 1$\footnote{This may not be true for a large
  $\delta$. Thus for $\delta>1$ we use $\gamma := \frac{1}{120 \bmax
    \Gamma^{2} \sum_{i=1}^n \frac{1}{i^{\amax-1}}}$.}.  Let us now
prove a bound on the probability that a close-by node transmits, which
is also required for the main theorem of this section.

\begin{lemma}
  \label{lem:1}
  Given an arbitrary node $v$. The probability that no node in the
  proximity region transmits in a given time slot is at least $1/4$.
\end{lemma}
\begin{proof}
  Let $3\Rmax(v)$ denote the set of nodes that are closer to $v$ than
  $3\Rmax$ in this argument. This is the set of nodes in the proximity
  region of $v$. The probability that a node in $3\Rmax(v)$ transmits
  in a single time slots is
\begin{align*}
  \label{eq:2}
  P_\text{none}^{3\Rmax(v)} \geq \prod_{u \in 3\Rmax(v)} \left( 1-p_u
  \right) \geq \left( \frac{1}{4} \right) ^{\sum_{u \in
      3\Rmax(v)} p_u} \geq \left( \frac{1}{4}
  \right)^{49\Gamma^2\cdot\gamma} \geq \left(
    \frac{1}{4} \right),
\end{align*}
where the second inequality holds due to
Fact~\ref{fact:1over4probabilitiesless1over2} from
Appendix~\ref{app:facts}, the third inequality due to a simple
geometric argument about the number of independent nodes within
distance $3\Rmax$ of $v$ and the bound on the sum of transmission
probabilities from within each transmission region. The last
inequality holds since $49\Gamma^2\cdot\gamma < 1$.
\end{proof}

Let us now consider nodes that are not in the proximity region of the
transmitting node. In order to bound the interference originating from
these nodes, we use rings around the transmitting node and bound the
probabilistic interference from within each ring. Note that although
our definition of the proximity region and rings differ, similar
arguments are made, for example, in \cite{hm-ttblb-12,gmw-lbpim-08}.
\begin{definition}
  For a node $v$, the \emph{ring} $C_i^v$, $i \geq 0$, is defined as
  the set of nodes with distance at least $(i+1) \cdot \Rmax$ and at
  most $(i+2) \cdot \Rmax$. For a ring $C_i^v$, the \emph{extended
    ring} $C_{i+}^v$ is defined as the set of nodes with distance at
  least $i \cdot \Rmax$ and at most $(i+3) \cdot \Rmax$.
\end{definition}
Note that for a ring $C_i^v$, the extended ring $C_{i+}^v$ is defined
such that the nodes in the transmission region of an arbitrary node $w
\in C_i^v$ are contained in $C_{i+}^v$.  If it is clear to which node
$v$ the rings refer, we write $C_i$ and $C_{i+}$ for brevity.

\begin{theorem}
  \label{thm:sum-of-rings-bounded-interference}
  Let the sum of transmission probabilities from each transmission
  region be upper bounded by $\gamma$. Given a node $v$, the
  probabilistic interference from nodes outside the proximity region
  of $v$ is upper bounded by $(\delta-1)N$.
\end{theorem}

\begin{proof}
  Let us first bound the interference from a single ring $C_i$. By a
  simple geometric argument it holds that the maximal number of
  independent nodes in the extended ring $C_{i+}$ is at most $(6i + 9)
  \Rmax^2/\Rmin^2$. By combining this number with the sum of
  transmission probabilities from within each broadcasting region, we
  can bound the interference from the nodes in $C_i$. As each node in
  the ring $C_i$ has distance greater than $i\cdot\Rmax$ from any node
  in $B_v$, it follows that the probabilistic interference on any node
  $u \in B_v$ is at most
  \begin{align*}
    \Psi_{C_i} &\leq \sum_{w \in C_{i+}} \frac{p_w \power_w }{ (i
      \Rmax)^{\amax} } 
    & \leq \frac{4 (6i + 9) \Rmax^2 \gamma \bmax \nmax}{ \Rmin^2 i^{\amax} }
    \cdot \left( \frac{ \Rmax }{ \Rmax
      } \right)^{\amax} 
    &\leq \frac{60 \gamma \bmax \nmax}{i^{\amax -1}} \cdot \left( \frac{ \Rmin
      }{ \Rmax } \right)^{2}.
  \end{align*}
  Summing over all rings it follows
  \begin{align*}
    \Psi_{w \not\in 3\Rmax(v)} & \leq \sum_{i=2}^\infty \Psi_{C_i} \leq
    60 \gamma \bmax \nmax \Gamma^2 \sum_{i = 1}^{n}
    \frac{1}{i^{\amax-1}} \leq \frac{(\delta-1)\nmax}{2},
  \end{align*}
  where the second inequality holds by inserting the bound on
  $\Psi_{C_i}$ and the fact that there are at most $n$ non-empty
  rings. The last inequality follows from the upper bound on $\gamma$,
  stated in Equation~\ref{eq:1}.
\end{proof}

\section{Local Broadcasting}
\label{sec:local-broadcasting}

In the previous section we have shown how to bound the probabilistic
interference from nodes outside of the proximity region based on an
upper bound on the sum of transmission probabilities from within each
transmission region. Such bounds are known for many algorithms in the
case of uniform transmission power, and hence we can plug our results
into a large body of related work, and transfer results with minimal
additional efforts to the case of arbitrary but fixed transmission
power.  In the following section we briefly state our results
regarding local broadcasting along with proof sketches as required. 
In Section \ref{sec:vari-transm-power} we discuss our results regarding
variable transmission power.

\subsection{Arbitrary but Fixed Transmission Power}
\label{sec:local-broadc-with-arb-tx-power}

The current results on local broadcasting with the knowledge of
$\Delta$ are based on transmitting with a fixed probability in the
order of $1/\Delta$ for a sufficient number of time slots in
$\O(\Delta \log n)$, while results that do not assume the maximal
degree $\Delta$ to be known are usually based on a so-called
slow-start mechanism. %Using such a mechanism, nodes start with very
%low transmission probability and double this probability until it is
%at least ``high enough''.%, followed by a decrease of the transmission
%probability once too many messages are received to avoid congestion.

\subsubsection{With knowledge of the maximal degree $\Delta$}
\label{sec:general-case}

Let us first consider the case, in which each node knowns the maximal
degree $\Delta$. Using the result on local broadcasting by
Goussevskaia, Moscibroda and Wattenhofer \cite{gmw-lbpim-08}, it is
easy to show that local broadcasting can be realized in
$\O(\Gamma^{2} \Delta \log n)$ time slots by simply adapting
the transmission probability to our requirements. 

\begin{theorem}
\label{thm:communication-ok-arb} 
Let the transmission probability of each node be $p = \gamma/\Delta$,
and $c>1$ an arbitrary constant.  A node $v$ that transmits with
probability $p$ for $8c/p \log n = \O(\Gamma^{2} \Delta \log n)$ time
slots successfully transmits to its neighbors whp.
\end{theorem}

\begin{proof}
  Since the transmission probability is chosen such that the sum of
  transmission probabilities from within each proximity range is at most
  $\gamma$, we can directly apply
  Theorem~\ref{thm:sum-of-rings-bounded-interference}.  Using the
  theorem, combined with the standard Markov inequality, the probability
  that the interference from nodes outside of the proximity region is
  too high (i.e., higher than $(\delta-1)\nmax$) is less than
  $1/2$. Lemma~\ref{lem:1} states that the probability that no node
  within the proximity range of a node transmits is greater than
  $\frac{1}{4}$. Combining both probabilities with the transmission
  probability of $p$ implies that the probability of a successful
  broadcast is at least $p/8$ in each time slot.  Thus transmitting
  for $8c/p \log n$ time slots results in a successful local broadcast
  with probability at least $1-\frac{1}{n^c}$. A detailed proof can be
  found in Appendix~\ref{sec:putt-piec-togeth}.
\end{proof}

\subsubsection{Without knowledge of $\Delta$}
\label{sec:lb-without-knowl-delta}

Let us now consider the case that the nodes are not given a bound on
the maximum degree $\Delta$.  In contrast to the previous algorithm
for local broadcasting, the ``optimal'' transmission probability is
initially unknown.

In order to create local broadcasting algorithms for this model, a
slow start mechanism can be used
\cite{hm-ttblb-12,yhwl-aodaa-12,ywhl-dlbap-11,gmw-lbpim-08}. In such a
mechanism each node starts with a very low transmission probability in
the range of $\O(1/n)$ and doubles the probability until a certain
number of transmissions are received, and the probability is reset to
a smaller value.  With such a mechanism, local broadcasting in the
(uniform-powered) SINR model can be achieved in $\O(\Delta\log n +
\log^2 n)$ \cite{hm-ttblb-12,yhwl-aodaa-12}. Although different forms
of the slow start mechanisms are used they reset the transmission
probabilities such that the sum of transmission probabilities in each
transmission region can be upper bounded by a constant.

Let us now consider the algorithm of Halld\'{o}rsson and Mitra, described
in \cite{hm-ttblb-12}. We can adapt the algorithm so that local
broadcasting provably works with high probability in the more general
model considered in this paper. This can be done by modifying the
maximal transmission probability to be $\gamma/16$ instead of $1/16$,
which can be done by simply changing Line 7 of Algorithm 1 in
\cite{hm-ttblb-12} from $p_y \leftarrow \min\{\frac{1}{16}, 2p_y\}$ to
$p_y \leftarrow \min\{\frac{\gamma}{16}, 2p_y\}$.  This minimal
adaptation allows us to bound the sum of transmission probabilities
similar to how it is done in the original paper. % A sketch of the proof is in
% Appendix~\ref{sec:app-lemma-towardstightbounds}
\begin{lemma}
  \label{lem:sending-prob-}
  Let $\mathcal{N}$ be a network with arbitrary transmission power
  assignment, asynchronous node wake-up and let all nodes execute
  Algorithm 1 from \cite{hm-ttblb-12} with maximal transmission
  probability be set to $\gamma/16$.  Then the sum of transmission
  probabilities from within each proximity region is upper bounded by
  $\gamma$.
\end{lemma}
By combining this result with
Theorem~\ref{thm:sum-of-rings-bounded-interference},
Lemma~\ref{lem:1}, and a similar argumentation as in the previous
section, the transmission is successful at least once with high
probability.  The correctness of the algorithm follows with the
original argumentation in \cite{hm-ttblb-12}.  Using the modified
Algorithm 1 from \cite{hm-ttblb-12}, we get for the more general case
of arbitrary transmission power assignment
\begin{theorem}
\label{thm:towards-tight-bounds-theorem}
There exists an algorithm for which the following holds whp: Each node
$v$ successfully performs a local broadcast within $\O((\Delta + \log
n) \Gamma^{2} \log n)$.
\end{theorem}

\textit{Remark:} Note that the local broadcasting algorithm by Yu
\emph{et. al.}  \cite{yhwl-aodaa-12} has the same runtime guarantees
as the algorithm by Halld\'{o}rsson and Mitra \cite{hm-ttblb-12}, but
was proposed slightly earlier. However, their algorithm cannot be
transfered to the case of arbitrary transmission power as is heavily
relies on bidirectional communication to operate. Specifically, their
algorithm computes an MIS, acquires information about dominated nodes
and then assigns transmission intervals to the dominated nodes. Thus,
it requires (at least) significant changes to generalize it to
networks of arbitrary transmission power.

\subsection{Variable transmission power}
\label{sec:vari-transm-power}

For local broadcasting, the transmission power is required to be fixed
for at least one full round of local broadcasting. In this section, we
consider a more general setting and allow the nodes to change the
transmission power for each time slot.  As it is not initially clear
which nodes should be considered as intended receivers in such a
setting, our result states the achieved broadcasting range, based on
the number of times certain transmission power levels were exceeded
within the considered time interval.  Note that we assume $\Delta$ to
be known to the nodes in this section.  We shall now briefly discuss
the notation required in this section.
% Let $\power^{(i)}_v$ be the transmission power of $v$ in time slot
% $i$.
We consider the time slots in one interval $(1,\dots,t)$. For multiple
time intervals that are not continuous, a transmission power of $0$
can be added to fill the gaps. Let $\{0=\power^{[0]}_v,
\power^{[1]}_v, \dots, \power^{[k]}_v\}$ the set of transmission
powers used by $v$ (plus 0), such that $\power^{[j]}_v <
\power^{[j+1]}_v$ for $j=0,1,\dots,k-1$.  We denote the number of time
slots, $v$ used a transmission power of at least $\power^{[j]}_v$ by
$T_j$.  Let $\R_v^{[j]}$ be the broadcasting range corresponding to
$\power^{[j]}_v$.

\begin{theorem}
  \label{thm:lb-variable-tx-power}
  Let all the nodes in the network transmit with probability at most
  $p = \gamma/\Delta$ and a variable transmission power between
  $\Rmin$ and $\Rmax$. Let $v$ be an arbitrary node that transmits
  with variable transmission powers during the interval
  $(1,\dots,t)$. For $j$ maximal such that $T_j > 8c/p \log n$, all
  nodes closer to $v$ than $\R_v^{[j]}$ received $v$'s message whp for
  an arbitrary constant $c>1$.
\end{theorem}

\begin{proof}
  Let $j$ be maximal such that $T_j > 8c/p \log n$.  It holds that $v$
  transmits with probability $p$ and transmission power at least
  $\power^{[j]}_v$ in at least $8c/p \log n$ time slots. Let us
  consider such a time slot $i$.  As the sum of transmission
  probabilities from within each proximity range is obviously bounded
  by at most $\gamma$, we can apply our method to bound the
  interference. It holds due to
  Theorem~\ref{thm:sum-of-rings-bounded-interference} and
  Lemma~\ref{lem:1} that a message transmitted by $v$ in time slot $i$
  is received by nodes closer to $v$ than $\R_v^{[j]}$ with
  probability at least $1/8$. Combined with the transmission
  probability $p$ and considered over $8c/p \log n$ time slots, this
  results in a success probability of at least $1-\frac{1}{n^c}$ with
  an argumentation similar to the that in the proof of
  Theorem~\ref{thm:communication-ok-arb} in
  Appendix~\ref{sec:putt-piec-togeth}.
\end{proof}

\section{Distributed Node Coloring and MIS}
\label{sec:node-coloring-mis}

We shall demonstrate the applicability of our results to existing
algorithmic results in the uniform SINR model in this
section. Therefore we consider a distributed node coloring
algorithm\cite{dt-dncsi-10}, and show how this algorithm can be
transfered to the case of arbitrary transmission powers. Distributed
node coloring is a fundamental problem in wireless networks, as a node
coloring can be used to compute a schedule of transmissions by
assigning each color to a different time slot.  Thus, efficient
transmissions based on a \emph{time-division-multiple-access} (TDMA)
schedule can be reduced to a node coloring.  The algorithm we consider
computes a node coloring that ensures that two nodes with the same
color cannot communicate directly.  This does not necessarily result
in a transmission schedule that is feasible in the SINR model,
however, one can use additional techniques like those described in
\cite{dt-dncsi-10} or \cite{fw-olbsc-13} to transform such a node
coloring to a local broadcasting schedule that is feasible in the SINR
model.  Let us now define some notation required for the coloring
problem. For two nodes $v,u \in V$ we say that there is a
\emph{communication link} from $v$ to $u$ if $u$ is in the
broadcasting region of $v$. We say that there is a
\emph{unidirectional} communication link from $v$ to $u$ if there is a
communication link from $v$ to $u$, but not from $u$ to $v$. In this
case $v$ \emph{dominates} $u$. If both communication links are
available we say that it is \emph{bidirectional}.  We call two nodes
$u$ and $v$ \emph{independent} if there is no communication link
between $u$ and $v$. Accordingly, a set is independent if each two
nodes in the set are mutually independent. A \emph{node coloring} is
\emph{valid} if each color forms an independent set.

Before stating the algorithms, we shall briefly characterize the
communication graph implied by arbitrary transmission powers in the
SINR model. Obviously, it is still based on a disk graph, but, not a
unit disk graph as in the uniform case. Additionally, there are two
main characteristics that are introduced by directed communication
links and are relevant for graph-based algorithms in this
setting. First, unidirectional communication links can form long
directed paths. This is formalized in the following definition.
\begin{definition} 
  Given a network $N$ and the induced communication graph
  $G=(V,E)$. Let $G'$ be the graph that remains after deleting all
  bidirectional edges from $G$. The \emph{longest directed path} in
  the network is defined as the longest simple path in $G'$. We denote
  the length of the longest directed path in a network by $\ell$.
\end{definition} 
Second, these directed paths cannot form a directed circuit. This
holds since in any circle in the communication graph, there must be a
bidirectional communication link. Consider a directed path consisting
of the nodes $(v_1,\dots,v_\ell)$.  It holds that the transmission
range decreases monotonically, i.e., $\Rmax_{v_i} \geq
\Rmax_{v_{i+1}}$ for $i = 1,\dots,\ell-1$. If a node $v_i$ can be
reached from $v_j$ with $i \leq j$, there must be a bidirectional
communication link as $v_i$ reaches $v_j$ as well due to $\Rmax_{v_i}
\geq \Rmax_{v_{j}}$.

\subsection{The Coloring Algorithm}
\label{sec:coloring-algorithm}

Let us now state the coloring algorithm.  The core of our algorithm is
based on the coloring algorithm by Moscibroda and Wattenhofer designed
for unstructured radio networks in \cite{mw-curn-05, mw-curn-08}. It
has been adapted to the case of uniform transmission powers in the
SINR model by Derbel and Talbi in \cite{dt-dncsi-10}. In this section
we extend the algorithm to work in the case of arbitrary transmission
power assignments. A state diagram of the algorithm can be found in
Figure~\ref{fig:coloring-algo}, and pseudocode of the states of the
algorithm can be found in Algorithms~\ref{algo:main} -
\ref{algo:colored}. Note that some technical details regarding the
wake-up of nodes and the impact on the algorithm are omitted here and
in the state diagram for simplicity, but are discussed in
Section~\ref{sec:asynchr-node-wakeup}.

We will now give an overview over the algorithm. The algorithm starts
with a three-way handshake protocol called neighborhood learning.
This allows each node to learn which of its incoming edges are
effectively bidirectional communication links. After this learning
stage, we allow a node $v$ to participate in the (modified) coloring
algorithm only if $v$ is not \emph{dominated}, i.e., if there is no
other uncolored node $w$ such that $w$ reaches $v$ but $v$ does not
reach $w$.

\begin{figure*}[hbt] 
  \centering
  \includegraphics[width=1\textwidth]{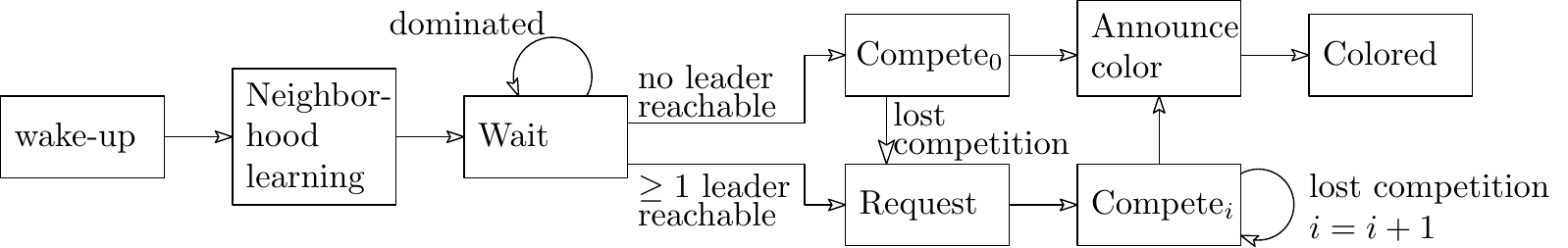}
  \caption{State diagram of the MW-coloring algorithm.}
  \label{fig:coloring-algo}
\end{figure*}

The coloring algorithm for node $v$ starts with a listening phase,
which is long enough so that $v$ knows the current status of all other
nodes that are awake and can reach $v$. Afterwards, if there is a
leader $w$ to which bidirectional communication is possible, $v$
enters the request state, and requests a color from $w$. After $w$
answers the request by assigning a color $j$, $v$ tries to verify the
assigned color $j$. If this is not successful (i.e., $v$ loses against
another node competing for $j$ that reaches $v$), $v$ increases $j$ by
one and retries. If $v$ is successful, it announces its success until
all the nodes that can hear $v$ are informed about $v$'s status and
hence know that $v$ will color itself with color $j$.

If there is no leader that can communicate bidirectionally with $v$,
$v$ tries to compete for the status leader. If this is not successful,
$v$ enters the request state (and proceeds as above) as there is a
leader with bidirectional communication available now. Note that $v$
does not lose against leader nodes that dominate $v$ as $v$ cannot
request a color from them. If $v$ is successful in becoming leader, it
selects a free leader color and announces its choice
so that all nodes that can be reached by $v$ are informed. 
After the announcement phase, the node is officially colored and will
only periodically transmit its color and serve color requests as they
arrive. 
Note that we will call the main coloring states of the algorithm (Compete,
Request, Announce and Colored) the \emph{core-coloring algorithm} in
this section.

Our presentation of some parts of the algorithms is based on the
presentation in \cite{dt-dncsi-10}. The pseudocode for the algorithm
can be found in Algorithms~\ref{algo:main} - \ref{algo:colored}. Let
$\chi(P_v)$ in Algorithm~\ref{algo:compete} be maximal such that
$\chi(P_v) \not \in \{d_v(w)-\zeta_i, \dots, d_v(w)+\zeta_i\}$ for
each $w\in P_v$ and $\chi(P_v) \leq 0$ and $t_1$ in
Algorithm~\ref{algo:colored} be the number of rounds a newly colored
node has to wait before serving color requests.

\begin{algorithm}[h!]
  \caption{Distributed node coloring in non-uniform power networks for
    node $v$}
  \label{algo:main}
  \begin{algorithmic}[1] %
    \State [ In$(v)$, Out$(v) ] =$ Neighborhood
    learning()%    

    \State Let the node store each received leader color ($M_V^i(w)$
    for $i \leq \Delta+1$) for $\kappa_s$ time slots in
    C$_\text{taken}$ 

    \State Listen for $(38\Gamma^2 + 3) \kappa_s$ time slots 
    
    \While {$c_v = -1$} \Comment{This is state Wait} \label{algo1:main-loop}%
    %\Comment{This runs at most $\O(\ell)$ times}%
    \If { (In$(v) \backslash$Out$(v)) \backslash $C$_\text{taken}$
      $ = \emptyset$}  \label{algo:main:notdominated}%
    \If { (In$(v) \cap $Out$(v) \cap $C$_\text{taken} = \emptyset$}
    \State Transition to Algorithm~\ref{algo:compete}: Compete$_0()$
    \Else 
    \State Request($w$) for a $w \in $ In$(v) \cap $Out$(v) \cap $C$_\text{taken}$
    \EndIf
    \Else %
    \State wait for $\kappa_s$ time slots%
    % \State if msg from a node $w$ in Incoming$(v)$ is received that
    % $w$ is colored, remove $w$ from Incoming$(v)$
    \EndIf%
    \EndWhile%
\end{algorithmic}
\end{algorithm}

\begin{algorithm}[h!]
  \caption{Neighborhood-Learning() for node $v$}
  \label{algo:neighborhood-learning}
  The neighborhood learning algorithm is a simple three-way-handshake
  protocol as introduced in \cite{t-ssn-75}. A node $v$ starts the
  neighborhood learning algorithm and thus sends a learning request
  with its own ID with probability $p_s$ for $\kappa_s$ time
  slots. For each reply it receives (at most $\Delta$), $v$ itself
  will acknowledge the reception.
\end{algorithm}

% \begin{algorithm}[h!]
%   \caption{Listen() for node $v$}
%   \label{algo:listen}
%   \begin{algorithmic}[1] 
%     \State Listen and mark colors in the range $1 \dots \Delta+1$ that
%     are already taken for $\kappa$ time slot
% \end{algorithmic}
% \end{algorithm}

\begin{algorithm}[h!]
  \caption{Compete$_i$() for node $v$}
  \label{algo:compete}
  \begin{algorithmic}[1] 
    \State The algorithm is based on the presentation in
    \cite{dt-dncsi-10}.  

  \State $P_v = \emptyset$, $\zeta_i = \begin{cases} \kappa_l &\mbox{if } i = 0 \\
    \kappa_s & \mbox{otherwise} \end{cases} $

    \State $\text{Next} = \begin{cases} \text{Request} &\mbox{if } i = 0 \\
      \text{Compete}_{i+1} & \mbox{otherwise} \end{cases} $ 

    \For {$\kappa_s$ time slots}
    \State \textbf{for each} $w \in P_v$ \textbf{do} $d_v(w) = d_v(w)
    + 1$
    \State \textbf{if} $M_A^i(w, c_w)$ rec. \textbf{then} $P_v =
    P_v \cup \{w\}$; $d_v(w) = c_w$
    \State \textbf{if} $M_C^i(w)$ rec. \textbf{then} goto Next; leader
    = $w$
    \EndFor

    \State $c_v = \chi(P_v)$

    \While{$true$}
    \State $c_v = c_v + 1$
    \If {$c_v > \kappa_s$}
    \State Announce$_j$() for $j \begin{cases} \text{minimal} \not \in
      $ C$_\text{taken} &\mbox{if } i = 0 \\
     i & \mbox{otherwise} \end{cases}$
    \EndIf

    \State \textbf{for each} $w \in P_v$ \textbf{do} $d_v(w) = d_v(w)
    + 1$
    \State transmit $M_A^i(v,c_v)$ with probability $p_s$
    \State \textbf{if} $M_C^i(w$ rec. \textbf{then} goto Next; leader
    = $w$
    \If {$M_A^i(w, c_w)$ rec. } 
    \State$P_v = P_v \cup \{w\}$; $d_v(w) = c_w$
    \State \textbf{if} $|c_v - c_w| \leq \zeta_i$ \textbf{then}
    $c_v = \chi(P_v)$
    \EndIf
    \EndWhile
\end{algorithmic}
\end{algorithm}

\begin{algorithm}[h!]
  \caption{Request($w$) for node $v$}
  \label{algo:request}
  \begin{algorithmic}[1] 
    \State Transmit $M_R^v(w)$ with prob. $p_s$
    for $\kappa_s$ slots to leader $w$ 
    \State Wait for $\kappa_s$ rounds to receive color assignment $j$.
    \State Transition to Algorithm~\ref{algo:compete} (Compete) with $i = j$
\end{algorithmic}
\end{algorithm}

\begin{algorithm}[h!]
  \caption{Announce$_i$() for node $v$}
  \label{algo:announce}
  \begin{algorithmic}[1] 
    \State Transmit the $M_C^i(v)$ announcement with $p_l$($p_s$) for
    $\kappa_l$($\kappa_s$) slots for leader / non-leader colors
    \State Wait \& Transmit $M_C^i(v)$ with $p_s$ for $\kappa_s$ slots
    \State Goto Algorithm~\ref{algo:colored}: Colored$_i$()
\end{algorithmic}
\end{algorithm}

\begin{algorithm}[h!]
  \caption{Colored$_i$() for node $v$}
  \label{algo:colored}
  \begin{algorithmic}[1] 
    \State count $= 0$, current $= -1$, serveCount $= 1$
    \While {$true$} 
    \State count = count$ + 1$
    \State Transmit $M_C^i(v)$ with probability $p_s$
    \State \textbf{if} $M_R^w(v)$ received \textbf{then} $Q = Q$.add($w$)
    \If {count $> t_1$ and $i \leq 9\Gamma^2+1$} 
    \State \Comment{Only for leader nodes that serve requests}
    \If {count $> t_1+\kappa_l$ or current $= -1$}
    \State count $= t_1$, serveCount = serveCount $+ 1$
    \State $j = $serveCount $\cdot 38\Gamma^2$.
    \State current = $Q$.first$()$ (or -1 if $Q$ empty)
    \EndIf
    \If {current $\not = -1$}
    \State Transmit $M_S^\text{current}(j)$ for $\kappa_l$ slots with
    prob. $p_l$, where $j$ is a free color (interval)
    \EndIf
    \EndIf
    \EndWhile
\end{algorithmic}
\end{algorithm}

% Let us briefly summarize the changes to the algorithm required
% \begin{enumerate}
% \item Bidirectional communication is not guaranteed, hence
%   neighborhood learning is required to ensure whether a node is
%   dominated or not. 
% \item The (modified) coloring routine is executed only if a node is not
%   dominated, instead of just executing it once.
% \item The transmission probabilities are adapted to fit the model and the
%   increased complexity of the algorithm.
% \item A color $i$ is only taken by a node $v$ after an announcement
%   phase. This allows a leader node $u$ dominated by $v$ that also selected
%   $i$ to retreat, which is required for asynchronous wake-up.
% \end{enumerate}

In order to allow leaders faster communication, the algorithm uses two
different transmission probabilities. Let the transmission probability
commonly used by non-leader nodes be $p_s = \gamma/(2\Delta)$ and the
transmission probability reserved for special leader tasks (i.e.,
announcement of winning a leader competition or answering color
requests) $p_l = \gamma/(18\Gamma^2)$.

% (with a transmission probability decrease by $\Gamma^2$ and $\Delta
% \Gamma^2$ in contrast to the bounds required for local broadcasting
% for leader and non-leader nodes, respectively\footnote{This leads to
%   an additional factor of $\Gamma^2$ in the runtime})

\subsection{Analysis}

Let us now begin with the analysis of the algorithm, which is split in
two parts. The first part shows that the transmissions conducted in
the algorithm are successful with high probability. In the second part
we will show that the algorithm computes a valid $\O(\Gamma^2\Delta)$
coloring, and terminates after at most $\O((\ell + \Delta)
\Gamma^{6} \Delta \log n)$ time slots.

\subsubsection{Transmissions are successful}
\label{sec:coloring-transm-success}

In order to apply the bound in the interference shown in
Section~\ref{sec:bound-interf}, we need to bound the sum of sending
probabilities from within each transmission region. 

\begin{lemma}
  \label{lem:coloring-leaders-bounded}
  Let $v$ be an arbitrary leader node. Then there are at most
  $9\Gamma^2$ other leader nodes in the transmission range of $v$.
\end{lemma}

\begin{proof}
  Note that leaders do not necessarily form an independent set, as an
  unidirectional communication link between leaders is allowed in the
  algorithm.
  Let us consider an arbitrary leader node $v$. It holds that within
  distance $\Rmin$ there cannot be another leader, as otherwise there
  would be a bidirectional communication link between two leader
  nodes. This is not possible as one of them would not have become
  leader but requested a color from the other. Thus it holds that for
  discs of size $\Rmin/2$ around each leader node in $v$'s
  neighborhood, these discs do not intersect. Hence it holds that
  there can be at most $\frac{(\Rmax+\Rmin/2)^2}{(\Rmin/2)^2} \leq
  9\Gamma^2$ leader nodes in a maximal transmission range.
\end{proof}

\begin{lemma}
  \label{lem:coloring-sending-prob-ok}
  Let leader nodes send with probability $p_l$ and non-leader nodes
  with probability $p_s$, then the sum of transmission probabilities
  from within each transmission region is upper bounded by $\gamma$.
\end{lemma}

\begin{proof}
  Let us consider an arbitrary node $v$ and sum over the transmission
  probabilities from within $v$'s transmission region 
  \begin{align*}
    \sum_{w \in B_v} p_w \leq 9\Gamma^2p_l + \Delta p_s \leq \gamma
  \end{align*}
  This holds as at most $9\Gamma^2$ leader nodes from each
  transmission region may transmit with probability $p_l$ due to
  Lemma~\ref{lem:coloring-leaders-bounded}, while at most $\Delta$
  other nodes in $v$'s neighborhood transmit with probability at most
  $p_s$.
\end{proof}
The corollary follows from the lemma along with the argumentation for
Theorem~\ref{thm:communication-ok-arb}. It shows that the limited
number of leader nodes are able to communicate to their neighbors in
$\O(\log n)$ time slots, while non-leader nodes require $\O(\Delta
\log n)$ time slots. Overall it implies that all transmissions in the
algorithm are successful w.h.p.
\begin{corollary}
  \label{cor:coloring-message-transmission-times}
  A message that is transmitted with probability $p_l$ ($p_s$) for
  $\kappa_l = 8c/p_l \log n$ ($\kappa_s = 8c/p_s \log n$) time slots reaches its intended
  receivers w.h.p.  
\end{corollary}

This shows that communication is successful with high probability even
in this more general case. Combined with the algorithmic changes and
the refined analysis in the full version of this paper
\cite{fw-atps-14}, the modified MW-coloring algorithm computes a
coloring with $\O(\Gamma^2 \Delta)$ colors such that each color forms
an independent set in $\O((\Delta+\ell)\Gamma^{4} \Delta \log n)$ time
slots.  This highlights the applicability of our method to bound the
interference in networks of nodes with arbitrary transmission powers.

\subsubsection{Runtime of the algorithm}
\label{sec:coloring-algo-runtime}

In this section we consider the runtime of the distributed node
coloring algorithm. We will first state the main result of this
section.
 
\begin{theorem}
  \label{thm:coloring-runtime}
  After running the coloring algorithm (Algorithm~\ref{algo:main}) for
  at most $\O((\Delta+\ell)\Gamma^{4} \Delta \log n)$ time
  slots, all nodes are colored.
\end{theorem}

The proof follows from the lemmata stated in this section and a worst
case execution of the algorithm. Let us therefore consider such a
worst case. The nodes starts with executing the neighborhood learning
protocol. Afterwards it will be dominated for the maximal time. Then,
finally the will be able to start running the core coloring
algorithm. It will therefore enter Compete$_0$ state, fail to win and
hence enter Request state afterwards. After going through the maximum
number of Compete$_i$ states, it will finally win a competition and
move (through announce) to the coloring state. Summing over the
maximal runtime of the states shows the theorem.
\begin{table}
\centering
\caption{Runtime of the algorithm. CC stands for parts of the core coloring
  algorithm}
\label{tab:coloring-runtime} 
\begin{tabular}{|l|l|l|} 
\hline
\textbf{State} & \textbf{Runtime} & \textbf{Proof}\\ \hline
Neighborhood learning & $(2\Delta+1) \kappa_s)$ & Lemma~\ref{lem:coloring-neighborhood-learning} \\ \hline
%Listen & $(38\Gamma^2+3)\kappa_s$ & - \\ \hline 
Wait state & $\ell \cdot $ Coloring & Lemma~\ref{lem:bounding-dominated-time}\\ \hline 
CC: Compete$_0$ & $3\kappa_s + \Delta \kappa_l$ & Lemma~\ref{lem:coloring-compete-0}\\ \hline 
CC: Compete$_i$ & $(38\Gamma^2+3)\kappa_s$ & Lemma~\ref{lem:coloring-compete-i}\\ \hline 
CC: Max. Compete$_i$'s & $(38^2\Gamma^4 +120\Gamma^2)\kappa_s$ & Lemma~\ref{lemma:coloring-bounded-competition-states}\\ \hline 
CC: Request & $(38+4)\kappa_s + \Delta\kappa_l$ & Lemma~\ref{lem:coloring-request}\\ \hline 
CC: Announce & $\leq 2\kappa_s$ & - \\ \hline 
%Coloring: Total Core& $\O(\Gamma^4 \kappa_s)$ & Lemma~\ref{lem:coloring-total-core}\\ \hline
\end{tabular}
\end{table}
In the following we prove the results stated in
Table~\ref{tab:coloring-runtime}. In the next lemma the runtime of the initial
neighborhood learning protocol bounded from above.

\begin{lemma}
  \label{lem:coloring-neighborhood-learning}
  Let a node $v$ execute
  Algorithm~\ref{algo:neighborhood-learning}. After the execution,
  both $v$ and its neighbors $w_i$ know about their communication link
  and whether the link is bidirectional.  The algorithm finishes
  within $2\Delta+1 \kappa_s$ time slots.
\end{lemma}

\begin{proof}
  As node $v$ transmits a neighborhood learning request to all its
  neighbors $w_i$, the neighbor answers within $\Delta \kappa_s$
  slots after receiving the request (he may serve at most $\Delta-1$
  other request in the meantime. If the $w_i$ reaches $v$, $v$
  receives the message and completes the three-way-handshake by
  acknowledging the reception of the message to $w_i$, again within at
  most $\Delta \kappa_s$ slots. This holds for all neighbors.
\end{proof}

After finishing the initialization, each node needs to wait until it
is no longer dominated. In the following we will argue that the core
coloring algorithm needs to run at most $\O(\ell)$ times before all
nodes are colored.

\begin{lemma}
  \label{lem:bounding-dominated-time}
  Each node $v$ that reached Line~\ref{algo1:main-loop} of
  Algorithm~\ref{algo:main} and is not dominated, will be colored
  after $\O(\Gamma^4\kappa_s)$ time slots.
\end{lemma}

\begin{proof}
  The runtime of the different states of the core coloring algorithms
  are as depicted in Table~\ref{tab:coloring-runtime}. Let us now
  assume a node $v$ is not dominated and in the required loop. It will
  then start executing the core coloring algorithm.  The worst case
  runtime of the core coloring algorithm is $\O(\Gamma^4\kappa_s)$
  time slots. This follows from the
  argumentation after Theorem~\ref{thm:coloring-runtime}.
\end{proof}

Specifically, the lemma implies that once a node reached
Line~\ref{algo1:main-loop} of Algorithm~\ref{algo:main}, it will be
colored after at most $\O(\ell \Gamma^4\kappa_s)$ time slots. This
holds since initially the length of the longest directed chain of
dominating nodes is $\ell$. Due to
Lemma~\ref{lem:bounding-dominated-time}, the length of the longest
\emph{uncolored} directed path is at most $\ell -1$ after
$\O(\Gamma^4\kappa_s)$ time slots.  After repeating this procedure for
$\ell$ times, the length of the longest uncolored directed path is $0$
and hence there are no dominated nodes.  Thus after one more execution
of the core coloring algorithms all nodes are colored.  Let us now
consider the states of the core-coloring algorithm. We begin with the
compete states for leader and non-leader nodes. 

\begin{lemma}
  \label{lem:coloring-compete-0}
  Let $v$ be a node entering the Compete$_0$ state. At most $3\kappa_s
  + \Delta\kappa_l$ slots after entering Compete$_0$, $v$ leaves the
  state.
\end{lemma}

\begin{proof}
  There are two cases. Either $v$ wins the competition and will become
  leader or loses and enters the request state afterwards. In both
  cases the initial listen stage takes $\kappa_s$ slots. However, as
  soon as $v$ transmits once (which is after at most another
  $\kappa_s$ slots), he cannot be reset anymore according to
  Lemma~\ref{lem:coloring-compete-no-reset}. Hence either $c_v$
  reaches $\kappa_s$ and $v$ becomes leader or another node reaches
  $\kappa_s$ first (with sufficient time before $c_v$ reaches
  $\kappa_s$) and hence forces $v$ in the request state.  As the
  counter $c_v$ may at worst be reset to $\Delta\kappa_l$, the overall
  runtime of state Compete$_0$ is $3\kappa_l + \Delta \kappa_s$.
\end{proof}

\begin{lemma}
  \label{lem:coloring-compete-i}
  Let $v$ be a node entering the Compete$_i$ state. At most
  $(38\Gamma^2+3)\kappa_s$ slots after entering Compete$_i$, $v$
  leaves the state.
\end{lemma}

\begin{proof}
  The Lemma follows from an argumentation analog to that of
  Lemma~\ref{lem:coloring-compete-0}.
\end{proof}

If a non-leader fails to verify the color $i$ it got assigned from its
leader, it tries to verify $i+1$ and so on. Thus non-leader nodes may
be in more than one consecutive compete states. We will now bound the
number of consecutive compete states a node may be forced to visit
before being able to verify a color.

\begin{lemma}
  \label{lemma:coloring-bounded-competition-states}
  A node can only be in $38\Gamma^2$ consecutive compete state and
  leaves the last compete state at most
  $38^2\Gamma^4+114\Gamma^2+6)\kappa_s$ slots after entering the first..
\end{lemma}

\begin{proof}
  Let us consider a node $v$ that got color $j$ assigned by its
  leader. It hold that the node will try to verify $j$ or a
  consecutive color, until it wins a competition and enters the
  announce state.  Let us consider the number of nodes that could
  force $v$ to move on to the next color.  By
  Lemma~\ref{lem:coloring-bounded-active-competing-nodes}, this number
  is upper-bound by $38\Gamma^2$. Hence after at most $38 \Gamma^2$
  consecutive compete states all nodes that may compete with $v$ for
  the same color are colored and hence $v$ succeeds in the following
  competition round.
\end{proof}
After proving the bound in the runtime of the compete states, let us
consider the request state.
\begin{lemma}
  \label{lem:coloring-request}
  A node $v$ that enters the request state leaves the request state at
  most $(38+4)\kappa_s + \Delta\kappa_l$ time slots afterwards.
\end{lemma}

\begin{proof}
  The node $v$ first sends it's request in $\kappa_s$ slots, and
  subsequently will be served by its leader. As the leader may still
  be in the initial not-yet-serving-requests phase (see
  Algorithm~\ref{algo:colored}), it may require up to $(38\Gamma^2 +
  3)$ slots until the leader starts serving the requests. As the 
  leader can have at most $\Delta$ request, $v$ will be served at most
  $\Delta\kappa_l$ slots later.
\end{proof}

\subsubsection{Correctness of the algorithm}
\label{sec:coloring-algo-correct}

In order to ensure the correctness of the algorithm it remains to show
that the algorithm indeed computes a valid node coloring with at most
$\O(\Gamma^2\Delta)$ colors. 

\begin{theorem}
  \label{thm:coloring-colors-valid}
  The coloring algorithm (Algorithm~\ref{algo:main}) computes a
  coloring with $\O(\Gamma^2 \Delta)$ colors such that each color forms an
  independent set.
\end{theorem}
We will show the theorem in two steps. We will first show that indeed
each color forms an independent set and afterwards bound the number
colors used by the algorithm.
\begin{proof}
  Let us consider two nodes $u$ and $v$ that are colored with the same
  color $i$.  Let us first assume there is a bidirectional
  communication link between $u$ and $v$. If $u$ and $v$ competed for
  $i$ at the same time, $u$ and $v$ cannot finish within less than
  $\kappa_l$ (or $\kappa_s$) time slots. Thus let us assume $v$
  finished before $u$. Then, $v$ was able to announce $i$ to
  $u$ and force $u$ to move to another color or the request state.  If
  $u$ and $v$ did not compete at the same time, let $v$ be the node
  colored earlier. Again, $v$ was able to communicate to $u$ that
  it is colored with $i$ and thus prevented $u$ from verifying
  $i$. Note that at least once $v$ reached $u$ less than
  $\kappa_s$ time slots before $u$ finishes, hence $i$ is in
  C$_\text{taken}$ of $u$.
\end{proof}
Let us now bound the number of colors
\begin{proof}
  As $9\Gamma^2$ is an upper bound on the number of other leader nodes
  that can be in the transmission range of a leader node, this is the
  maximal number of colors that can be blocked when a leader node
  selects it's color. Hence $9\Gamma^2+1$ leader colors are
  sufficient.  The number of non-leader colors is bound by the number
  of requests a leader may have to serve in the worst case. This is
  obviously $\Delta$ as bidirectional communication is required. Due
  to Lemma~\ref{lemma:coloring-bounded-competition-states} it holds
  that for each request at most $38 \Gamma^2$ consecutive colors are
  required. After noticing that $38 \Gamma^2$ is the first non-leader
  color that is assigned it holds that at most $38 \Gamma^2 (\Delta+1)$
  non-leader colors are used by the algorithm.
\end{proof}

\subsubsection{Asynchronous node wakeup}
\label{sec:asynchr-node-wakeup}

Let us now briefly consider the asynchronous wake-up of nodes.  In
order to allow nodes to start after other nodes finished the
neighborhood learning protocol, we allow both algorithms to run in
parallel by requiring each node to reserve every second round of
``local broadcasting'' for answering possible neighborhood learning
requests. This requires to account twice the number of time slots for
each transmission, as well as each request. This doubles the runtime
of the algorithm, but enables the algorithm to cope with asynchronous
node wakeup.

We assume that two nodes that currently execute the core-coloring algorithm do
not have an unidirectional link. However, such an unidirectional link might be
introduced due to an awaking node.  To prevent this, we require nodes
that are not yet colored to stop executing the core-coloring algorithm
and return to the main loop of Algorithm~\ref{algo:main} immediately
if they get dominated. Note that colored leaders need to store which
colors they assigned to which nodes and reuse them accordingly in
order to ensure the bound on the number of colors in the previous
section. Note that if a node $v$ is already colored it is not required to
stop running the core coloring algorithm. However, as a node that has
a unidirectional communication link to $v$ selects the same color as
$v$, $v$ needs to resign from its color. 

Thus if a node $v$ that is colored with color $i$ receives an
announcement from node $w$ that $w$ will take color $i$, $v$ finishes
serving its requests and then resigns from the color and enters the
main loop of Algorithm~\ref{algo:main}. Note that $v$ cannot be in the
initial phase in which it is not yet allowed to serve
requests. Otherwise $w$ would have been in the core coloring algorithm
at the same time as $v$ and dominated $v$. Hence $v$ could not verify
the leader color $i$ (due to the listen-phase in Algorithm~\ref{algo:main}).

As we do also have to handle nodes that go to sleep asynchronously, we
require the nodes to enter the main loop of the algorithm if for
example a request is not answered within the time boundaries proven in
Section~\ref{sec:coloring-algo-runtime}.

Note that the runtime of the algorithm holds only for stable parts of
the network. As nodes that wake up may force other nodes to resign, we
cannot guarantee a runtime based only on the wake up time of the node
itself.
However, the runtime of Theorem~\ref{thm:coloring-runtime} holds for $v$
after the last node that can reach $v$ or one of $v$'s neighbors directly or
through a directed chain woke up. This holds as $v$ can only be forced
to stop the algorithm or resign from its color by a node that reaches
$v$ directly or through a directed chain. However, $v$ may expect
a delay for example in the request state only if a neighbor of $v$ is
forced to resign.

\subsection{Maximal Independent Set}
\label{sec:maxim-indep-set}

An algorithm for solving MIS can be deducted by simplifying our
coloring algorithm. As nodes can either be in the MIS or not, we do
only require two colors. Let 0 be to color that indicates that a node
is in the MIS and 1 that it is not. As all nodes in the MIS
are independent, we do not require the request state, and nodes in the
MIS do not need to serve requests. Also, once a node $v$ that is
executing the core-coloring algorithm receives a $M_C^0(w)$ message,
$v$ can instantly transition to the Colored$_1()$ algorithm.  After a
runtime of $\O((\Delta^2+\ell)\Gamma^{\amax + 4} \log n)$, each node
selected a color and thus either is in the MIS or not.

\section{Conclussion}
\label{sec:conclussion}

In this paper we have proven a bound on the interference in networks
with arbitrary transmission power assignments in wireless ad hoc
networks. We believe that this generic result will be of use in many
algorithms designed for such networks.  We have shown that local
broadcasting can be transfered to the general case of arbitrary
transmission powers with minor efforts due to this
result. Additionally, we considered variable transmission power, which
allows each node to change its transmission power in each time
slot. To highlight the applicability of our results on communication
in networks with arbitrary transmission power, we presented a
distributed node coloring algorithm that is fully adapted to
characteristics of directed communication networks such as
unidirectional communication links.  For future directions, we wonder
whether the dependence on the neighborhood learning algorithm is
required and whether the dependence on $\Gamma$ could be decreased.

\section*{Acknowledgments} This work was supported by
the German Research Foundation (DFG) within the Research Training
Group GRK~1194 "Self-organizing Sensor-Actuator Networks".

\bibliographystyle{abbrv}

{
\bibliography{../../bib/abbrv,../../bib/bib}
}
% ACM needs 'a single self-contained file'!
%\balancecolumns
\appendix
%Appendix A
%use subsection as highest section command.

\FloatBarrier
\section{Omitted Proofs}

\subsection{Proof of Theorem \ref{thm:communication-ok-arb}}
\label{sec:putt-piec-togeth}

We shall prove the theorem in three steps. First, we establish that
within the transmission radius of each node the transmission
probability is constant. Then, we consider the the probabilistic
interference that origins from the area close to the sender, and
finally we sum over the probabilistic interference from all nodes in
the network by exploiting that if not too many nodes transmit in any
part of the network the interference from further away parts are
negligible.
\begin{figure}[htb]
  \centering
  \includegraphics[trim=6.7cm 5.47cm 0cm 6.2cm, clip=true, width=0.45\textwidth]{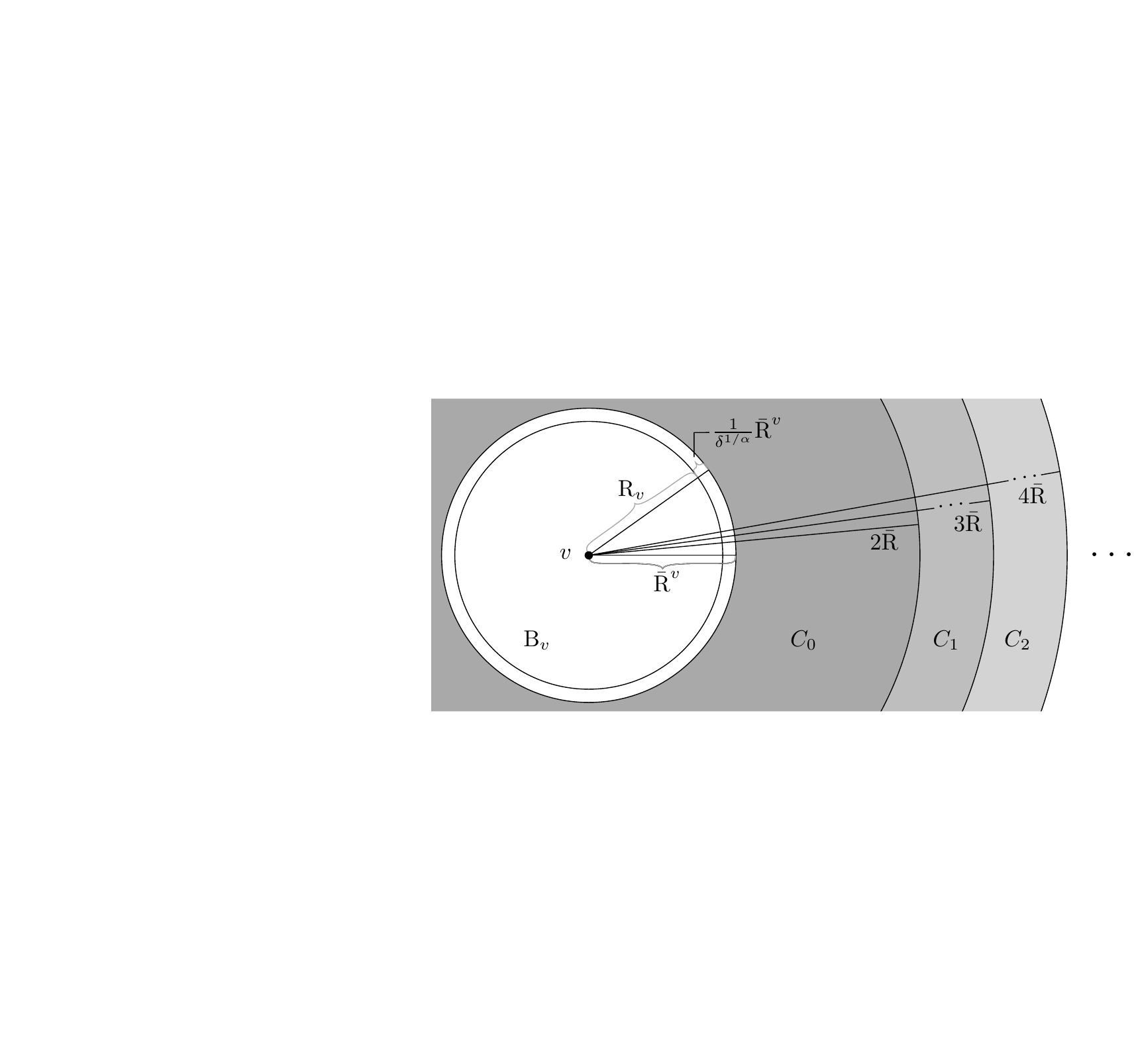}
  \caption{Sketch of proof}
  \label{fig:sketch-of-proof}
\end{figure}
We are now able to proof the result using standard techniques from
\cite{gmw-lbpim-08} combined with the results in Section~\ref{sec:bound-interf}.

\addtocounter{theorem}{-5}
\begin{theorem}
\label{thm:communication-ok-arb-2} 
Let the transmission probability of each node be $p = \gamma/\Delta$,
and $c>1$ an arbitrary constant.  A node $v$ that transmits with
probability $p$ for $8c/p \log n = \O(\Gamma^{2} \Delta \log n)$ time
slots successfully transmits to its neighbors whp.
\end{theorem}

\begin{proof}
  Let $v$ be the node that transmits for $8c/p \log n$ time slots with
  probability $p$. We prove the theorem by first showing that the
  probability of a successful local broadcast of $v$ in each round it
  transmits is substantial, followed by proving that at least one
  successful local broadcast of $v$ happens with high probability
  within $8c/p \log n$ time slots $v$

  It is stated in Theorem~\ref{thm:sum-of-rings-bounded-interference}
  that the probabilistic interference from all nodes not in
  $\Rmax^v$ is upper bounded by $(\delta-1)\nmax/2$. With the
  standard Markov inequality it follows that the probability that the
  interference from outside of the maximal transmission radius exceeds
  $(\delta-1)\nmax$ with probability less than $1/2$ and thus that
  the SINR condition holds with probability at least $1/2$. Combining
  both probabilities with the transmission probability of $v$ yields a
  lower bound on the success of a local broadcast by $v$ in each time
  slot.
  \begin{align*}
    P_\text{success} \geq \frac{p}{8}
  \end{align*}
  Using this probability we can bound the probability that $v$ fails
  in having a successful local broadcast within $8c/p \log n$ time
  slots. 
  \begin{align*}
    P_\text{fail} & \leq \left(1-\frac{p}{8}
    \right)^{8c/p \log n} = \left(1-\frac{c\log n}{8/p c\log n }
    \right)^{8c/p \log n}  \overset{(1)}{\leq} e^{-c\log n} = \frac{1}{n^c}, 
  \end{align*}
  where (1) follows from
  Fact~\ref{fact:1plustovernalltothepowerofn}. Hence within $8c/p \log
  n = O(\Gamma^{\amax+2} \Delta \log n)$ time slots, at least one of
  the transmissions of $v$ is successfully heard by all nodes in $B_v$
  with high probability.
\end{proof}

\subsection{Coloring}
\label{sec:app-coloring-proofs}

The following lemmata are required to ensure that the counters in
compete states are not reset for ever, but that some nodes will be
able to reach the counter limit and thus get colored.

\begin{lemma}
  \label{lem:coloring-counter-bounds}
  For the counter value $c_v$ of a node $v$ in the compete state it
  holds that $c_v \geq -\Delta\kappa_l$ if $v$ is in state
  Compete$_0$, and $c_v \geq -(38\Gamma)\kappa_s$ if $v$ is in state
  Compete$_i$ for $i > 0$.
\end{lemma}

\begin{proof}
  The lemma follows directly from the argumentation for Lemma 5 in \cite{dt-dncsi-10}.
\end{proof}

\begin{lemma}
  \label{lem:coloring-compete-no-reset}
  For a node $v$ in the compete state it holds that once he
  successfully transmitted a message with its counter value $c_v$ to
  its neighbors, it cannot be reset anymore.
\end{lemma}

\begin{proof}
  The lemma follows directly from the argumentation for Lemma 6 in \cite{dt-dncsi-10}.
\end{proof}

\subsubsection{Bounding the number of compete states }
The following two lemmata are required to bound the number of consecutive
compete states.
\begin{lemma}
  \label{lem:coloring-bounded-leaders-in-range}
  Let $v$ be an arbitrary node. Then at a given time slot at most
  $19\Gamma^2$ leader nodes can be within a distance of $2\Rmax$ of
  $v$.
\end{lemma}

\begin{proof}
  With the same argumentation as in the proof of
  Lemma~\ref{lem:coloring-sending-prob-ok}, it holds that discs with
  radius $\Rmin/2$ around the leader nodes do not intersect, and are
  fully contained in a disc of radius $2\Rmax + \Rmin/2$ around $v$.
  Thus at most $\frac{\area (2\Rmax+\Rmin/2)}{\area (\Rmin/2)} \leq
  19\Gamma^2$ nodes can be leaders within distance $2\Rmax$ around
  $v$.
\end{proof}

\begin{lemma}
  \label{lem:coloring-bounded-active-competing-nodes}
  Given a network with asynchronous wake-up of nodes.  Let $v$
  be an active node that tries to verify the assigned color $j$. Then
  there are at most $38\Gamma^2$ nodes $u_1,\dots,u_c$ that are active,
  capable of communicating with $v$, and that try to verify the same
  color $j$ as $v$.
\end{lemma}

\begin{proof}
  Let us consider a time slot $t$ such that $t$ is the first time slot
  in which a node $v$ competes with more than $38\Gamma^2$ nodes for
  the same non-leader color $j$.  Let us denote the upper bound on the
  time it takes to compete for one color as $T$ for the sake of simplicity.  
  As $t$ is the first time slot, it holds that all nodes that received a color-assignment
  $38\Gamma^2 T$ time slots before $t$ or earlier must have finished competing for the
  considered color $j$.

  Due to Lemma~\ref{lem:coloring-bounded-leaders-in-range} at most
  $19\Gamma^2$ nodes can be leaders around $v$, and thus in a given
  period of at least $19\Gamma^2 T$ time slots at most $19\Gamma^2$
  nodes within distance $\Rmax$ of $v$ can get the same color assigned
  as $v$ (due to the listen period of $19\Gamma^2 T$ time slots before
  a new leader node answers requests).  Hence within the $38\Gamma^2
  T$ time slots before $t$, at most $38\Gamma^2$ nodes may compete for
  color $j$, contradicting the choice of $t$ and implying the lemma.
\end{proof}

\subsection{Useful facts}
\label{app:facts}

\begin{fact} (proven in
  \cite{js-pawp-02}) \label{fact:1over4probabilitiesless1over2} 

  \noindent 
  Given a set of probabilities $p_1,\dots,p_n$ with $\forall i : p_i
  \in [0,\frac{1}{2}]$, the following inequalities hold:
  \begin{align*} \left(\frac{1}{4}\right)^{\sum_{k=1}^n p_k} \leq
\prod_{k=1}^n (1-p_k) \leq \left(\frac{1}{e}\right)^{\sum_{k=1}^n p_k}
  \end{align*}
\end{fact}

\begin{fact} (for example in the mathematical background section of
  \cite{mp-ra-10}) \label{fact:1plustovernalltothepowerofn} 

  \noindent 
  For all $n$, $t$, such that $n \geq 1$ and $|t| \leq n$,
  \begin{equation*} e^{t} (1-\frac{t^2}{n}) \leq (1+\frac{t}{n})^{n}
\leq e^{t}.
  \end{equation*}
\end{fact}

\subsection{A Worst Case Network}
\label{sec:worst-case-network}

\begin{figure}[ht]
  \centering
  \includegraphics[width=4cm]{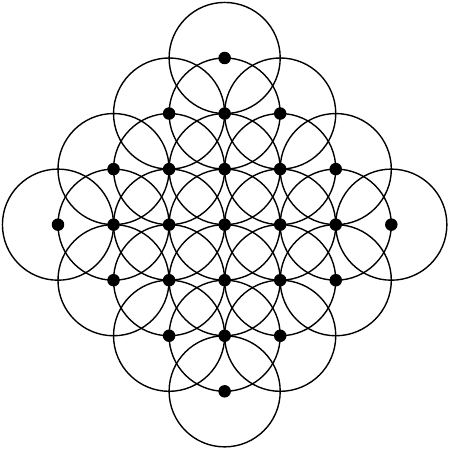}
  \caption{A network that requires $\Omega(n)$ time slots to allow each node one local broadcast if the broadcasting range equals the transmission range.}
  \label{fig:worst-case-omega-n}
\end{figure}

\end{document}